\documentclass[conference]{IEEEtran}
\IEEEoverridecommandlockouts
\usepackage{graphicx}
\usepackage{soul}
\usepackage{array}
\usepackage{amsmath}
\usepackage{float}
\usepackage{multirow}
\usepackage{array}
\usepackage{cancel}
\usepackage{tabu}
\usepackage{amsthm}
\usepackage{epsfig}
\usepackage{epstopdf}
\usepackage{epsf}
\usepackage{color}
\usepackage[english]{babel}

\theoremstyle{definition}
\newtheorem{theorem}{Theorem}
\newtheorem{lemma}{Lemma}
\newtheorem{definition}{Definition}

\usepackage{amssymb}
\usepackage{cite}
\usepackage{amsmath,amssymb,amsfonts}
\usepackage{graphicx}
\usepackage{textcomp}
\usepackage{xcolor}
\usepackage{algorithm}
\usepackage[noend]{algpseudocode}
\usepackage[T1]{fontenc}
\usepackage{graphicx}
\usepackage{lipsum}
\usepackage{afterpage}
\usepackage{enumitem}
\usepackage{lipsum}
\usepackage{mathtools}
\usepackage{cuted}

\usepackage{amsthm}
\theoremstyle{remark}

\usepackage[caption=false,font=footnotesize]{subfig}

\usepackage{algorithm}
\usepackage[noend]{algpseudocode}
\usepackage[hidelinks]{hyperref}

\usepackage[top=0.73in, bottom=1.1in, left=0.63in, right=0.63in]{geometry}

\usepackage{comment}
\usepackage{optidef}

\begin{document}

\title{Cross-Domain Lossy Compression via Constrained Minimum Entropy Coupling}
	
\author{\IEEEauthorblockN{Nam Nguyen}
\IEEEauthorblockA{School of Electrical Engineering\\ and Computer Science\\
Oregon State University\\
Corvallis, OR, 97331\\
Email: nguynam4@oregonstate.edu}
\and
\IEEEauthorblockN{Thinh Nguyen}
\IEEEauthorblockA{School of Electrical Engineering\\ and Computer Science\\
Oregon State University\\
Corvallis, OR, 97331 \\
Email: thinhq@eecs.oregonstate.edu}
\and
\IEEEauthorblockN{Bella Bose}
\IEEEauthorblockA{School of Electrical Engineering\\ and Computer Science\\
Oregon State University\\
Corvallis, OR, 97331 \\
Email: bella.bose@oregonstate.edu}
}

\author{\IEEEauthorblockN{$\textnormal{Nam Nguyen}$, $\textnormal{Hassan Tavakoli}$, $\textnormal{An Vuong}$, $\textnormal{Thinh Nguyen}$, and $\textnormal{Bella Bose}$ \thanks{This work was supported by the National Science Foundation under Grant No. CCF:SHF:2417898.}}
\IEEEauthorblockA{School of Electrical Engineering and Computer Science, Oregon State University, Corvallis, OR, 97331\\}
Emails: \{nguynam4, tavakolh, vuonga2, thinhq, bella.bose\}@oregonstate.edu}
	
\maketitle

\begin{abstract}
This paper studies cross-domain lossy compression through the lens of minimum entropy coupling (MEC) with rate and classification constraints. In this setting, an encoder observes samples from a degraded source domain, while the decoder is required to generate outputs following a prescribed target distribution and to preserve information relevant to a downstream classification task. Motivated by logarithmic-loss distortion, we adopt an information-based objective that maximizes the coupling strength between the source and reconstruction, rather than minimizing a sample-wise distortion. Under common randomness, we formulate a rate-constrained MEC problem (MEC-B) and show that the intermediate representation can be removed without loss of optimality, yielding an equivalent deterministic coupling formulation. For Bernoulli sources, closed-form expressions are derived with and without classification constraints. In addition, we implement a neural restoration framework using
quantization, entropy modeling, distribution matching, and classification regularization. Experiments on MNIST super-resolution and SVHN denoising show that increasing the available rate improves classification accuracy and yields more informative reconstructions.
\end{abstract}

\renewcommand\IEEEkeywordsname{Keywords}
\begin{IEEEkeywords}
Cross-domain lossy compression, minimum entropy coupling, image compression, image restoration, deep learning 
\end{IEEEkeywords}

\section{Introduction}

Rate-distortion theory provides the classical foundation for lossy compression
by characterizing the optimal tradeoff between coding rate and reconstruction
distortion~\cite{cover1999elements}. Traditional compression systems are often
evaluated using pointwise fidelity measures, such as mean squared error (MSE),
peak signal-to-noise ratio (PSNR), and structural similarity index
(SSIM)~\cite{wang2004image}. However, in many modern generative, restoration,
and task-oriented compression settings, pointwise distortion alone is not
sufficient. For instance, the encoder may observe degraded samples
\(X\sim p_X\), such as noisy or low-resolution images, while the decoder is
required to generate outputs that follow a different target distribution
\(p_Y\), such as clean or high-resolution images
\cite{liu2022lossy,nguyen2026crossdomain,ebrahimi2024mecb}. In addition, the
reconstruction should retain semantic information that is useful for downstream
tasks~\cite{blau2018perception,blau2019rethinking,agustsson2019generative}.

These challenges have led to several extensions of classical rate-distortion
theory. The rate-distortion-perception (RDP)
framework~\cite{blau2018perception,blau2019rethinking} studies the tradeoff
among rate, distortion, and perceptual quality, where perception is measured by
a divergence between the source and reconstruction distributions. Generative
compression methods based on generative adversarial networks (GANs) and
distribution-matching losses have shown that visually realistic
reconstructions can be produced at low rates, although often with higher
distortion~\cite{goodfellow2014generative,arjovsky2017wasserstein,
tschannen2018deep,larsen2016autoencoding}. More recently, task-aware
compression has incorporated classification performance into the compression
objective. In particular, rate-distortion-classification (RDC)
~\cite{CDP,Zhang2023,Wang2024,nguyen2026rdc,NamNguyen2025} and
rate-distortion-perception-classification (RDPC)
~\cite{nguyen2026crossdomain} characterize tradeoffs among rate, fidelity,
distributional quality, and classification performance. These frameworks are
mainly distortion-based, whereas this paper considers an information-based
criterion.

Logarithmic loss offers a natural information-theoretic alternative to
sample-wise distortion. It is widely used in prediction and statistical
learning, and has also been studied in source coding, including multiterminal
settings~\cite{courtade2011multiterminal,courtade2013multiterminal}. In soft
reconstruction, where the decoder outputs a probability distribution rather
than a deterministic estimate, expected log-loss is directly related to
conditional entropy and mutual information. Single-shot lossy source coding
under logarithmic loss was further investigated by Shkel and
Verd{\'u}~\cite{shkel2017single}. These connections motivate measuring
reconstruction quality through the mutual information preserved between the
source and reconstruction.

Minimum entropy coupling (MEC) gives another information-theoretic view of
distribution-constrained reconstruction. Given fixed marginals, MEC seeks a
joint distribution with minimum joint entropy. Since the marginal entropies are
fixed, this is equivalent to maximizing the mutual information between the
coupled variables. MEC has been studied
in~\cite{vidyasagar2012minimum,cicalese2019minimum} and is NP-hard in
general~\cite{kovacevic2015entropy,vidyasagar2012minimum}. This difficulty has
motivated approximation methods, including greedy algorithms for causal
inference~\cite{kocaoglu2017causal,kocaoglu2017entropic} and additive
approximation guarantees within a constant number of bits of the
optimum~\cite{compton2023minimum}.

The minimum entropy coupling with bottleneck (MEC-B) problem extends MEC by
introducing a rate-limited intermediate representation between the source and
the reconstruction~\cite{ebrahimi2024mecb}. This formulation is well-suited to
distribution-constrained lossy compression, where the reconstruction must
match a prescribed target marginal while being generated from a rate-limited
representation. Related work includes lossy compression with a distinct source
and reconstruction distributions via entropy-constrained optimal
transport~\cite{liu2022lossy}, as well as classification-constrained
cross-domain compression~\cite{nguyen2026crossdomain}. Unlike these
distortion-based approaches, this paper develops a log-loss/information-based
MEC framework for cross-domain lossy compression under rate and classification
constraints.

In this paper, we study cross-domain lossy compression through rate- and
classification-constrained MEC. Given degraded samples from \(p_X\) and desired
reconstructions from \(p_Y\), the goal is to maximize \(I(X;Y)\) while
enforcing the target marginal, a rate constraint, and a task constraint on the
uncertainty of a label \(S\) given the reconstruction \(Y\). Under common
randomness, we show that the discrete problem admits an equivalent
deterministic coupling representation and derive closed-form Bernoulli
characterizations with and without classification constraints. We further
implement a neural restoration framework using quantization, entropy modeling,
distribution matching, and classification regularization. Experiments on MNIST
super-resolution and SVHN denoising show that increasing the available rate
improves classification accuracy and yields more informative reconstructions.

\section{Problem Formulation}

We consider a lossy compression system described by the Markov chain $X \leftrightarrow Z \leftrightarrow Y,$ where the source $X \sim p_X$ is encoded into a representation $Z$ via a stochastic encoder $p_{Z|X}$ and reconstructed as $Y$ through a stochastic decoder $p_{Y|Z}$. The encoder is subject to a rate constraint $H(Z) \leq R$.

Instead of measuring fidelity using a sample-wise distortion function $d(x,y)$ \cite{cover1999elements}, we adopt a distributional criterion based on logarithmic loss. This leads to maximizing the mutual information $I(X;Y)$, which is equivalent to minimizing $H(X|Y)$. The log-loss distortion measure has been studied in the context of rate-distortion theory in \cite{courtade2011multiterminal,courtade2013multiterminal}, and is particularly suitable when the decoder produces soft (distribution-valued) reconstructions \cite{shkel2017single}.


To ensure meaningful reconstructions, we impose that the output follows a prescribed marginal distribution $p_Y$. This avoids degenerate solutions (e.g., identity mappings) and aligns the reconstruction with a desired target distribution. Following \cite{ebrahimi2024mecb}, we define the rate-constrained minimum entropy coupling problem (MEC-B):
\begin{align}
\label{eq:mecb_rewrite}
\mathcal{I}_{\mathrm{MEC\text{-}B}}(p_X,p_Y,R)
&=
\max_{p_{Z|X},\,p_{Y|Z}} I(X;Y) \\
\text{s.t.} \quad
& X \leftrightarrow Z \leftrightarrow Y, \nonumber \\
& H(Z) \leq R, \nonumber \\
& P_X = p_X, \nonumber \\
& P_Y = p_Y. \nonumber
\end{align}  
The constraint $P_Y = p_Y$ enforces that the reconstruction matches a desired target distribution, which is essential in applications such as image restoration and generative compression. The objective $I(X;Y)$ quantifies the strength of the coupling between $X$ and $Y$, capturing how much information about the source is preserved under the rate constraint.

\begin{figure}[]
\centering
\includegraphics[width=0.47\textwidth]{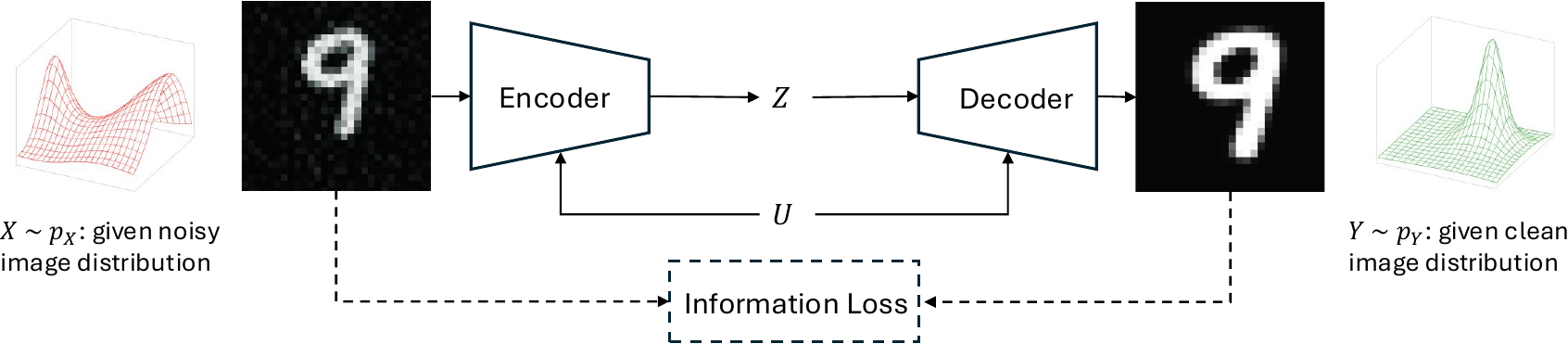}
\caption{System model: a noisy input $X \sim p_X$ is restored as $Y \sim p_Y$.}
\label{fig:System_Model_Cross_Domain_MEC}
\end{figure}

\section{MEC-B with Common Randomness}

\subsection{Under Rate Constraint}

We extend the formulation in \eqref{eq:mecb_rewrite} by incorporating shared randomness between the encoder and decoder. Let $U$ denote a random variable independent of $X$, available to both terminals as shown in Figure \ref{fig:System_Model_Cross_Domain_MEC}.

\begin{definition}
\label{def:mecb_randomness}
Define $M(p_X, p_Y)$ as the set of joint distributions $p_{U,X,Z,Y}$ with marginals $p_X, p_Y$ that factorize as
$p_{U,X,Z,Y} = p_U \, p_X \, p_{Z|X,U} \, p_{Y|Z,U}$, where $U$ is the shared common randomness. The minimum entropy coupling problem with rate constraint and shared randomness (MEC-B-R) is given by
\begin{align}
\label{prob:mecb_randomness}
   \mathcal{I}_\text{MEC-B-R}(p_X, p_Y, R) 
    &= \max_{p_{U,X,Z,Y} \in M(p_X, p_Y)} \, I(X;Y)  \\
    \text{s.t.} \quad
    & H(Z|U) \leq R. \nonumber 
\end{align}
\end{definition}
Conditioned on the shared randomness $U$, the encoder maps $X$ to a representation $Z \sim p_{Z|X,U}$, which can be compressed at an average rate $H(Z|U)$. The decoder, given $(Z,U)$, generates the reconstruction $Y \sim p_{Y|Z,U}$. By standard source coding arguments, the constraint $H(Z|U) \leq R$ characterizes the achievable rate.

\begin{figure}[ ]
\centering
\includegraphics[width=0.34\textwidth]{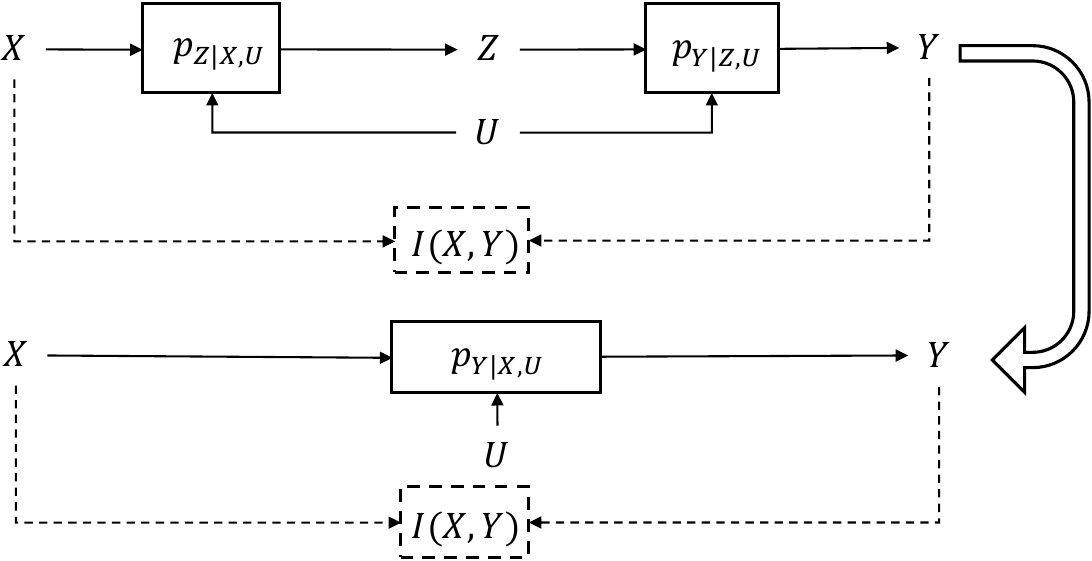}
\caption{System architecture corresponding to Theorem~\ref{thm:oneshot_random}. The encoder maps $(X,U)$ to $Y$, and the decoder reconstructs $Y$ using $U$.}
\label{fig:Architectures_MEC}
\end{figure}

The optimization in \eqref{prob:mecb_randomness} is carried out jointly over the distribution of $U$ and the stochastic mappings $p_{Z|X,U}$ and $p_{Y|Z,U}$. The following result provides an equivalent and more compact representation.

\begin{theorem}
\label{thm:oneshot_random} 
Define $Q(p_X, p_Y)$ as the set of joint distributions $p_{U,X,Y}$ with marginals $p_X, p_Y$ that factorize as $p_{U,X,Y} = p_U \, p_X \, p_{Y|X,U}$. Then, the minimum entropy coupling problem in Definition \ref{def:mecb_randomness} admits the representation
\begin{align}
\label{prob:mecb_randomness_deterministic_convert}
    \mathcal{I}_\text{MEC-B-R}(p_X, p_Y, R) 
    &= \max_{p_{U,X,Y} \in Q(p_X, p_Y)} \, I(X;Y)  \\
    \textnormal{s.t.} \quad 
    & H(Y|X,U) = 0,  \nonumber\\
    & I(X;U) = 0, \nonumber \\
    & H(Y|U) \leq R. \nonumber 
\end{align}
\end{theorem}

\begin{proof}
The result follows by adapting Theorem~3 in \cite{liu2022lossy}. 
For completeness, a detailed proof is provided in Appendix~\ref{app:oneshot_random_proof}.
\end{proof}

Theorem~\ref{thm:oneshot_random} shows that the intermediate representation $Z$ can be removed without loss of optimality. In particular, the reconstruction $Y$ can be generated directly from $(X,U)$, with the constraint $H(Y|X,U)=0$ ensuring that the mapping is deterministic given $U$.

Under this equivalent formulation, the encoder produces $Y$ from $(X,U)$ and compresses it losslessly at rate $H(Y|U)$, while the decoder reconstructs $Y$ using the shared randomness $U$. Thus, the problem reduces to designing a conditional distribution $p_{Y|X,U}$ subject to a rate constraint on $Y$.

\subsubsection{Bernoulli Case Expression}
We now specialize the problem in \eqref{prob:mecb_randomness_deterministic_convert} to the Bernoulli setting.
\begin{theorem}
\label{Bernoulli_radom_IRC}
Let \(X \sim \mathrm{Bern}(q_X)\) and \(Y \sim \mathrm{Bern}(q_Y)\) with
\(0 < q_X, q_Y \leq \tfrac{1}{2}\). Under common randomness, the
rate-constrained Bernoulli MEC-B problem is given by
\begin{align*}
\mathcal{I}_\text{MEC-B-R}^{(B)}(q_X, q_Y, R)
&=
H_b(q_Y) - (1-q_X)H_b\bigl(q_Y-q_X\alpha\bigr) \\
&-
q_X H_b\bigl(q_Y+(1-q_X)\alpha\bigr),
\end{align*}
where
\begin{align*}
\alpha
=
\begin{cases}
\min\left\{\dfrac{R}{H_b(q_X)},\, \dfrac{q_Y}{q_X}\right\},
& q_Y \le q_X,\\[2ex]
\min\left\{\dfrac{R}{H_b(q_X)},\, \dfrac{1-q_Y}{1-q_X}\right\},
& q_Y \ge q_X.
\end{cases}
\end{align*}
and $H_b(t)=-t\log t-(1-t)\log(1-t)$ denotes the binary entropy function.
\end{theorem}

\begin{proof}
The proof is shown in Appendix~\ref{app:Bernoulli_radom_IRC_proof}.
\end{proof}

\begin{figure}[ ]
\centering
\includegraphics[width=0.32\textwidth]{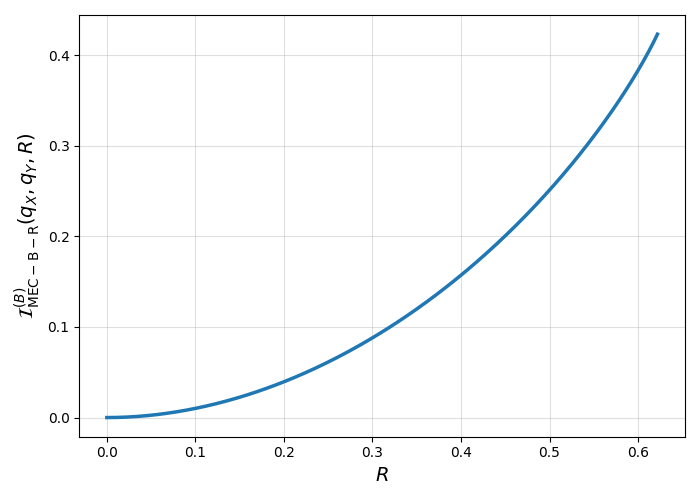}
\caption{$\mathcal{I}_\text{MEC-B-R}^{(B)}(q_X, q_Y, R)$ versus $R$ for $q_X = 0.2$ and $q_Y = 0.3$.}
\label{fig:I_R_Bernoulli}
\end{figure}

Figure~\ref{fig:I_R_Bernoulli} shows 
\(\mathcal{I}_\text{MEC-B-R}^{(B)}(q_X, q_Y, R)\) versus \(R\) for
\(q_X = 0.2\) and \(q_Y = 0.3\). The curve increases with \(R\), showing that
a larger bottleneck rate allows \(Y\) to preserve more information about \(X\)
while satisfying the prescribed marginal constraint. 

\begin{figure}[ ]
\centering
\includegraphics[width=0.47\textwidth]{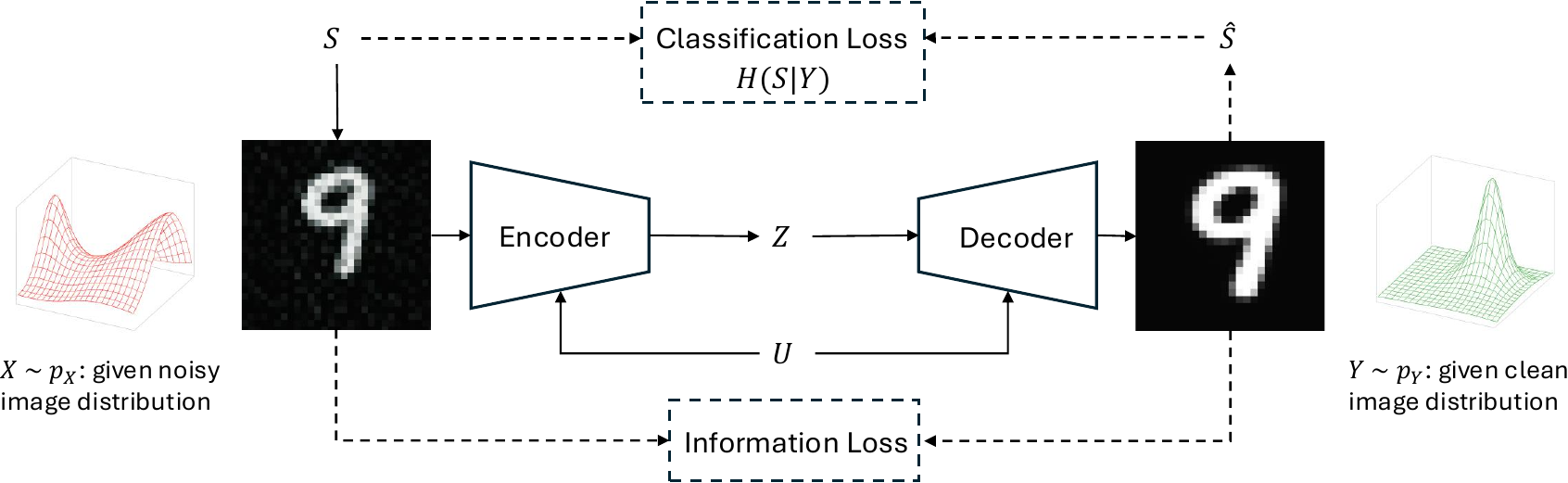}
\caption{System model: a noisy input $X \sim p_X$ is restored as $Y \sim p_Y$, while supporting classification with label $S$.}
\label{fig:System_Model_Cross_Domain_MEC_Classifiation}
\end{figure}

\begin{figure*}
\begin{align}
\label{eq:MEC_B_R_C}
&\mathcal{I}_\text{MEC-B-R-C}^{(B)}(q_X, q_Y, R, C)  \\ 
& = \begin{cases} 
    \displaystyle H_b(q_Y) - (1-q_X)H_b \left[q_Y-q_X\frac{R}{H_b(q_X)} \right] 
    - q_XH_b\left[q_Y+(1-q_X)\frac{R}{H_b(q_X)}\right], \\ \hspace{1cm} \frac{R (H_b(q_{S_1}) - H_b(m))} {H_b(q_X)}  + H_b(m) < C < H_b(m) \text{ and } R \leq H_b(q_X) \min\left\{
    \frac{q_Y}{q_X},
    \frac{1-q_Y}{1-q_X}
    \right\} \\
    \displaystyle H_b(q_Y) - (1-q_X)H_b\left[q_Y + q_X\frac{R}{H_b(q_X)}\right] 
    - q_XH_b\left[q_Y - (1-q_X)\frac{R}{H_b(q_X)}\right],\\ 
    \hspace{1cm} \frac{R (H_b(q_{S_1}) - H_b(m))}{H_b(q_X)}  + H_b(m) < C < H_b(m) \text{ and } R \leq H_b(q_X) \min\left\{
    \frac{q_Y}{1-q_X},
    \frac{1-q_Y}{q_X}
    \right\} \\
    \displaystyle H_b(q_Y) - (1-q_X)H_b\left[q_Y-q_X\frac{C - H_b(m)}{H_b(q_{S_1}) - H_b(m)}\right] 
    - q_XH_b\left[q_Y+(1-q_X)\frac{C - H_b(m)}{H_b(q_{S_1}) - H_b(m)}\right], \\
    \hspace{1cm} H_b(q_{s_1}) \leq C \leq \frac{R (H_b(q_{S_1}) - H_b(m))}{H_b(q_X)} + H_b(m)  \text{ and } C \geq H_b(m) - (H_b(m) - H_b(q_{S_1}))\min\left\{
    \frac{q_Y}{q_X},
    \frac{1-q_Y}{1-q_X}
    \right\} \\
    \displaystyle H_b(q_Y) - (1-q_X)H_b\left[q_Y + q_X\frac{C - H_b(m)}{H_b(q_{S_1}) - H_b(m)}\right] 
    - q_XH_b\left[q_Y - (1-q_X)\frac{C - H_b(m)}{H_b(q_{S_1}) - H_b(m)}\right], \nonumber\\
    \hspace{1cm} H_b(q_{s_1}) \leq C \leq \frac{R (H_b(q_{S_1}) - H_b(m))}{H_b(q_X)} + H_b(m) \text{ and } C \geq H_b(m) - (H_b(m) - H_b(q_{S_1}))\min\left\{
    \frac{q_Y}{1-q_X},
    \frac{1-q_Y}{q_X}
    \right\}\\
    \displaystyle H_b(q_Y)- (1-q_X)H_b\left(\frac{q_Y- \min\{q_X,q_Y\}}{1-q_X}\right) -
    q_XH_b\left(\frac{\min\{q_X,q_Y\}}{q_X}\right), \hspace{0.2cm}  H_b(q_S) < C < H_b(m) \text{ and } R > H_b(q_X).
\end{cases}
\end{align}    
\end{figure*}

\subsection{Under Rate and Classification Constraints}
\label{subsec:bernoulli_case_classification}

We extend the formulation in \eqref{prob:mecb_randomness} by incorporating constraints on task-relevant information. In addition to matching the target marginal $p_Y$, the reconstruction is required to preserve information about a target variable $S$ as shown in Figure \ref{fig:System_Model_Cross_Domain_MEC_Classifiation}.

\noindent\textbf{Classification constraint.}
In addition to reconstruction quality, we impose an explicit constraint on task-relevant information. Specifically, the reconstruction must satisfy $H(S|Y) \leq C$ for some $C > 0$. This constraint bounds the residual uncertainty of $S$ given $Y$, thereby ensuring a prescribed level of classification performance \cite{Wang2024,nguyen2026crossdomain,nguyen2026rdc,NamNguyen2025}. Equivalently, the condition $H(S|Y)\le C$ implies $I(S;Y) \ge H(S) - C$, ensuring that the reconstruction retains sufficient information for downstream inference tasks.

\begin{figure}[ ]
\centering
\includegraphics[width=0.32\textwidth]{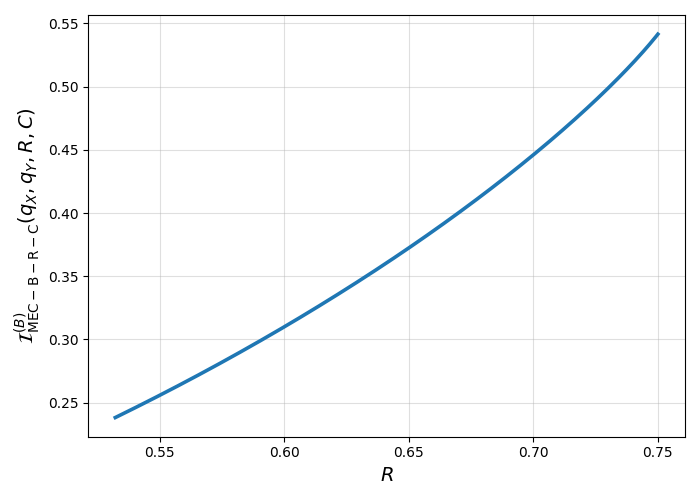}
\caption{$\mathcal{I}_\mathrm{MEC-B-R-C}^{(B)}(q_X,q_Y,R,C)$ versus $R$ for $q_X = 0.3$, $q_Y = 0.4$, $q_{S_1} = 0.01$, and $C=0.4$.}
\label{fig:I_R_Bernoulli_C}
\end{figure}

\begin{definition}
\label{def:DRC_random_oneshot_def_classification}
Define $M(p_X, p_Y)$ as the set of joint distributions $p_{U,X,Z,Y}$ with marginals $p_X, p_Y$ that factorize as $p_{U,X,Z,Y} = p_U \, p_X \, p_{Z|X,U} \, p_{Y|Z,U}$, where $U$ is the shared common randomness. The constrained MEC problem with rate constraint $R$ and classification constraint $H(S|Y)\le C$ under shared randomness is given by
\begin{align}
\label{prob:oneshot_random_def_classification}
     \mathcal{I}_\text{MEC-B-R}(p_X, p_Y, R, C)  
    &= \max_{p_{U,X,Z,Y} \in M(p_X, p_Y)} \, I(X;Y)  \\
    \text{s.t.} \quad & H(Z|U) \leq R, \nonumber\\
    & H(S|Y) \leq C. \nonumber 
\end{align}
\end{definition}

\begin{figure*}[ ]
\centering 
\includegraphics[width=0.58\textwidth]{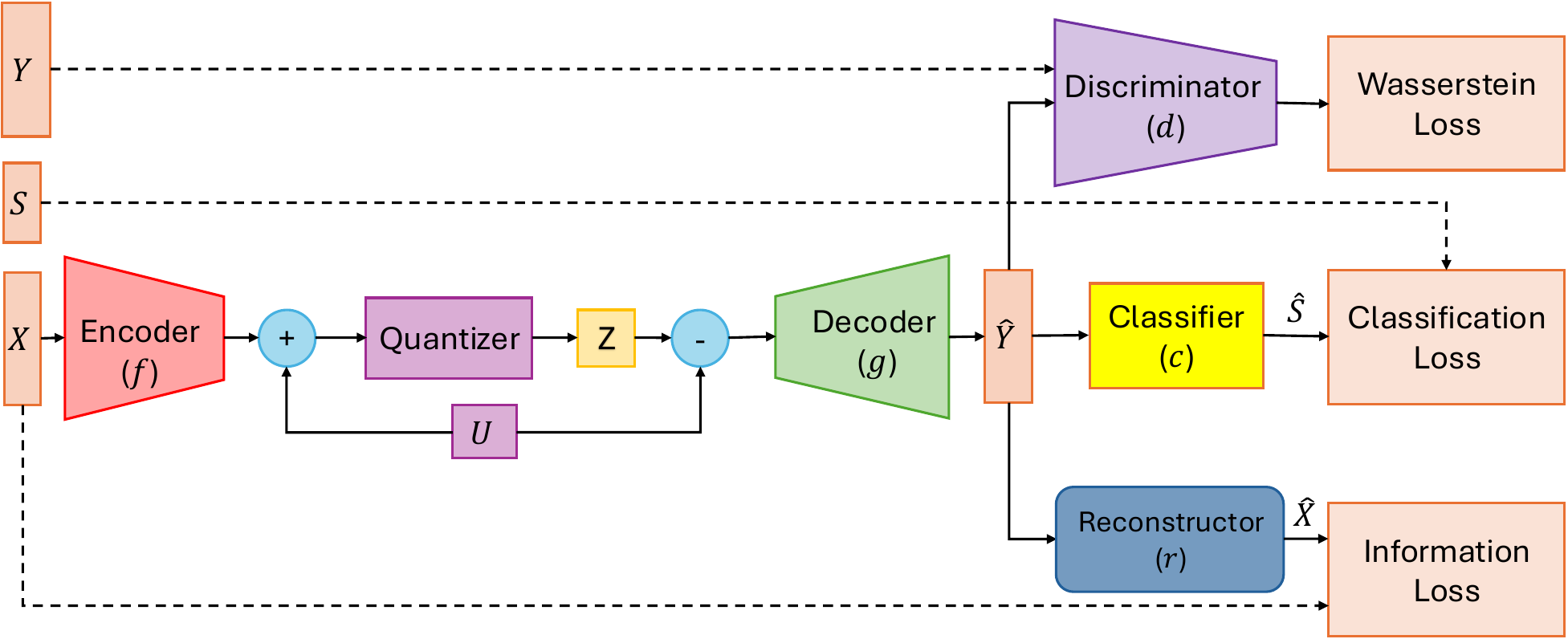}
\caption{Experimental architecture: a stochastic autoencoder with classifier, reconstructor, and WGAN discriminator, conditioned on shared randomness $U$.}
\label{fig:Scheme_Cross_Domain_MEC}
\end{figure*} 

\subsubsection{Bernoulli Case Expression}
We now specialize the formulation in  \eqref{prob:oneshot_random_def_classification} to the Bernoulli setting.  
Let $X \sim \mathrm{Bern}(q_X)$ and $Y \sim \mathrm{Bern}(q_Y)$ with $0 < q_X, q_Y \leq \tfrac{1}{2}$. Using $\oplus$ to denote modulo-2 addition, the classification variable is modeled as $S = X \oplus S_1$, where $S_1 \sim \mathrm{Bern}(q_{S_1})$ with $0 < q_{S_1} \leq \tfrac{1}{2}$. This induces the marginal distribution $q_S = P(S=1) = q_X + q_{S_1} - 2q_X q_{S_1}$.
\begin{theorem}
\label{Oneshot_Bernoulli_radom_IRC_classifiction}
The problem~(\ref{prob:oneshot_random_def_classification}) is feasible only if
$C\geq H_b(q_{S_1})$. Under common randomness, the
rate- and classification-constrained MEC-B-R-C problem for Bernoulli source admits the
closed-form solution of $\mathcal{I}_\text{MEC-B-R-C}^{(B)}(q_X, q_Y, R, C)$
as given in Equation~(\ref{eq:MEC_B_R_C}), where $m = (1 - q_X)(1 - q_{S_1}) +  q_X q_{S_1}$.
\end{theorem}

\begin{proof}
The proof is provided in Appendix~\ref{app:Oneshot_Bernoulli_radom_IRC_classifiction_proof}.
\end{proof}

Figure~\ref{fig:I_R_Bernoulli_C} shows
\(\mathcal{I}_\mathrm{MEC-B-R-C}^{(B)}(q_X,q_Y,R,C)\) versus \(R\) for
\(q_X = 0.3\), \(q_Y = 0.4\), \(q_{S_1}=0.01\), and \(C=0.4\). The curve
increases with \(R\), indicating that a larger bottleneck rate allows \(Y\) to
preserve more information about \(X\) while satisfying both the marginal and
classification constraints.

\section{Experimental Results}
\label{sec:Rate_Upper_Bound}

\subsection{Training Setup}

\begin{figure*}[ ]
\centering
\subfloat[Accuracy vs.~rate (MNIST).]{
    \includegraphics[width=0.28\linewidth]{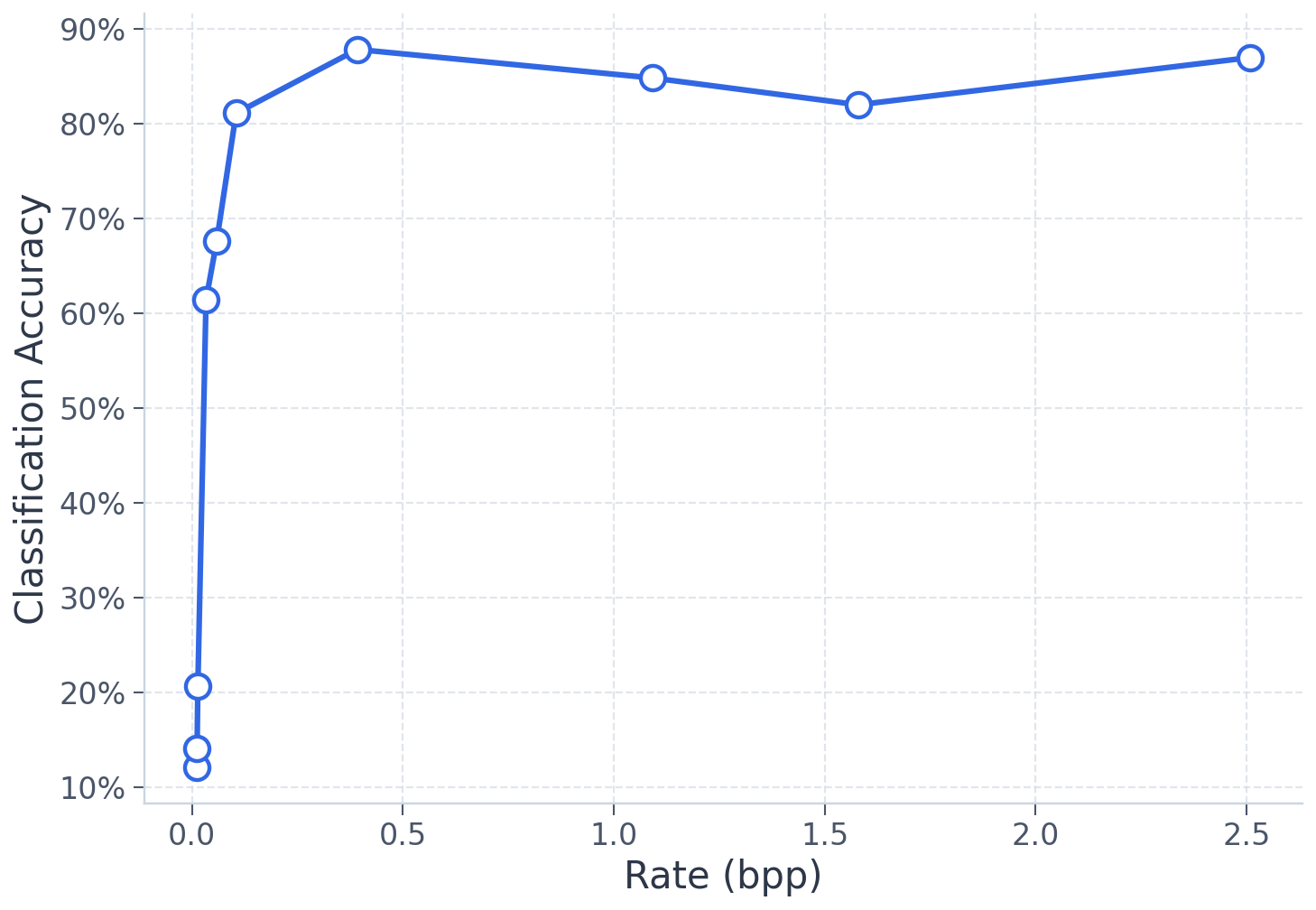}
    \label{fig:Accuracy_Rate_Resolution_MEC}
}
\hfill
\subfloat[Cross-entropy vs.~rate (MNIST).]{
    \includegraphics[width=0.28\linewidth]{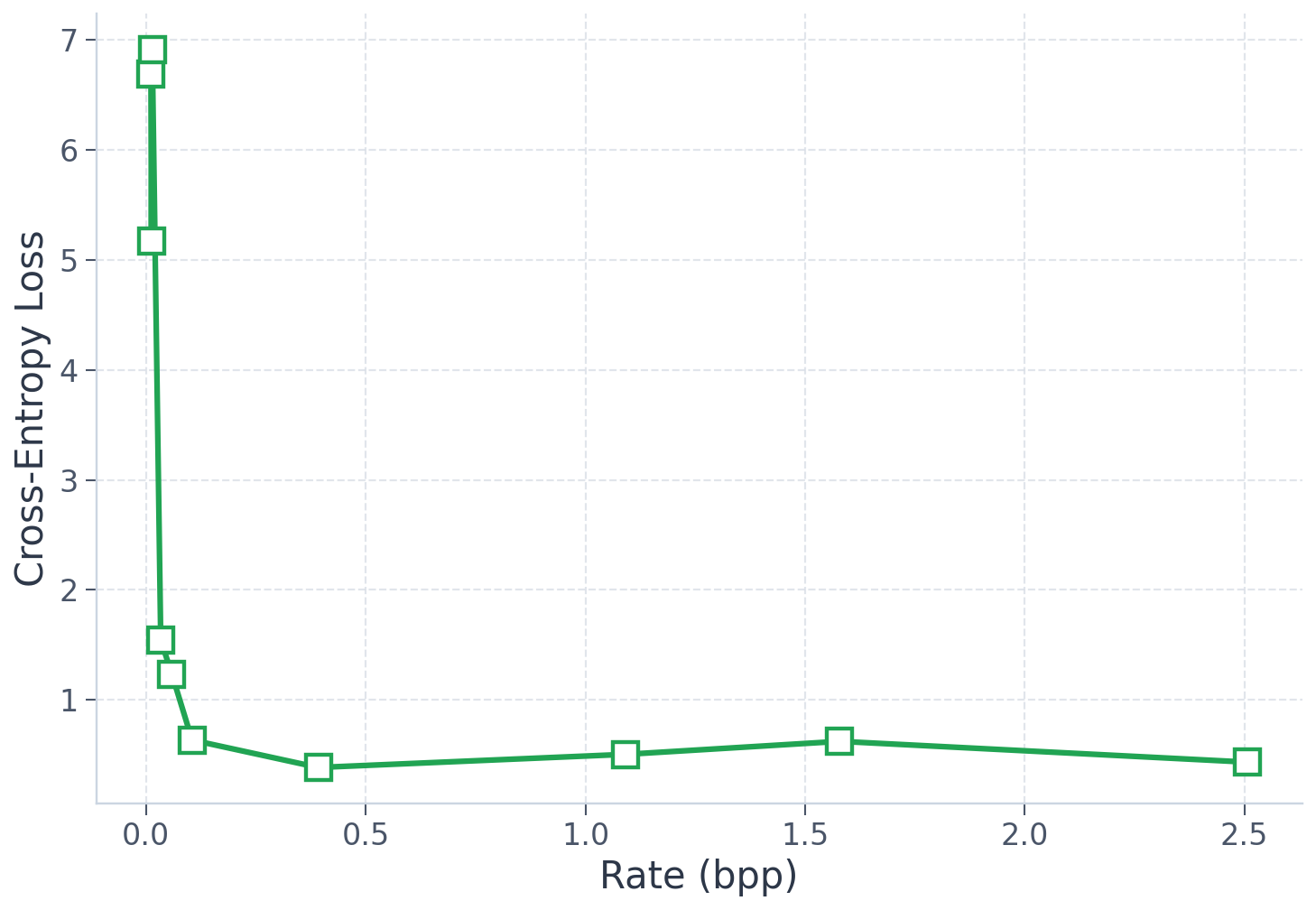}
    \label{fig:Cross_Entropy_Loss_Rate_Resolution_MEC}
}
\hfill
\subfloat[Reconstruction images (MNIST).]{
    \includegraphics[width=0.39\linewidth]{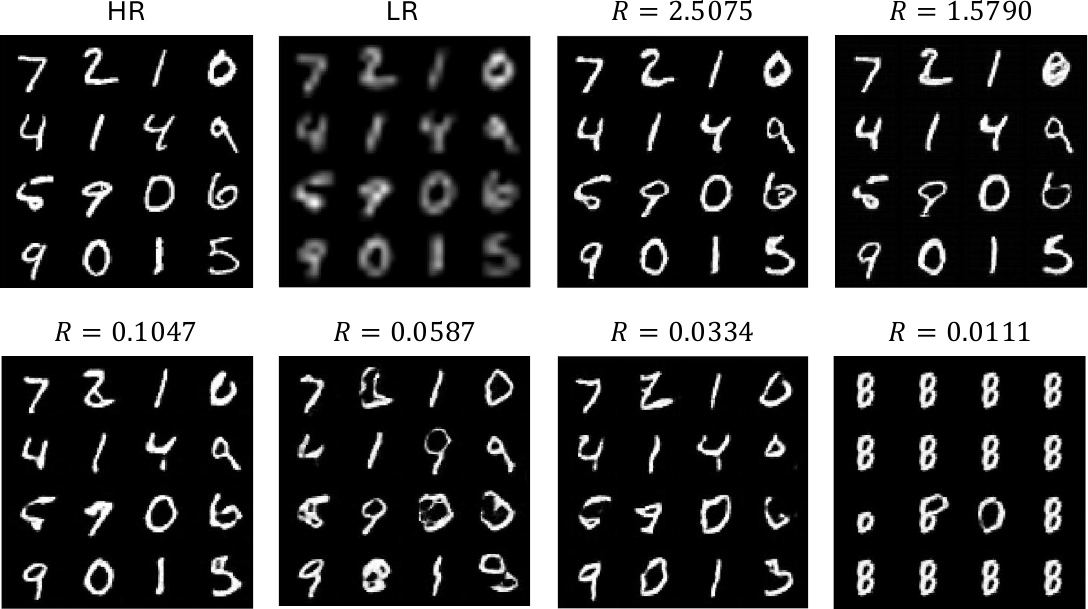}
    \label{fig:Reconstruction_Resolution_MEC}
}
\hfill
\subfloat[Accuracy vs.~rate (SVHN).]{
    \includegraphics[width=0.28\linewidth]{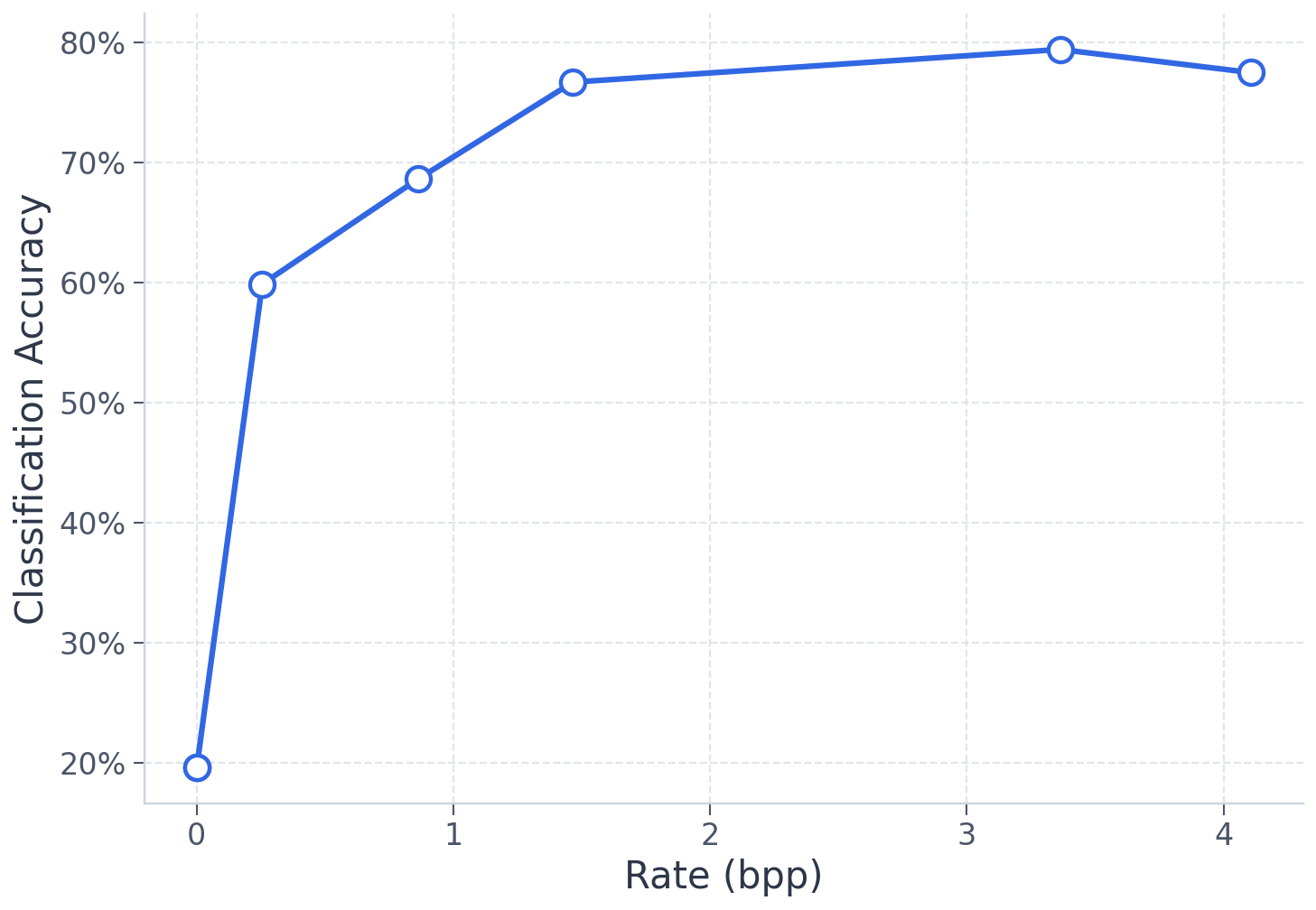}
    \label{fig:Accuracy_Rate_Denoising_MEC}
}
\hfill
\subfloat[Cross-entropy vs.~rate (SVHN).]{
    \includegraphics[width=0.28\linewidth]{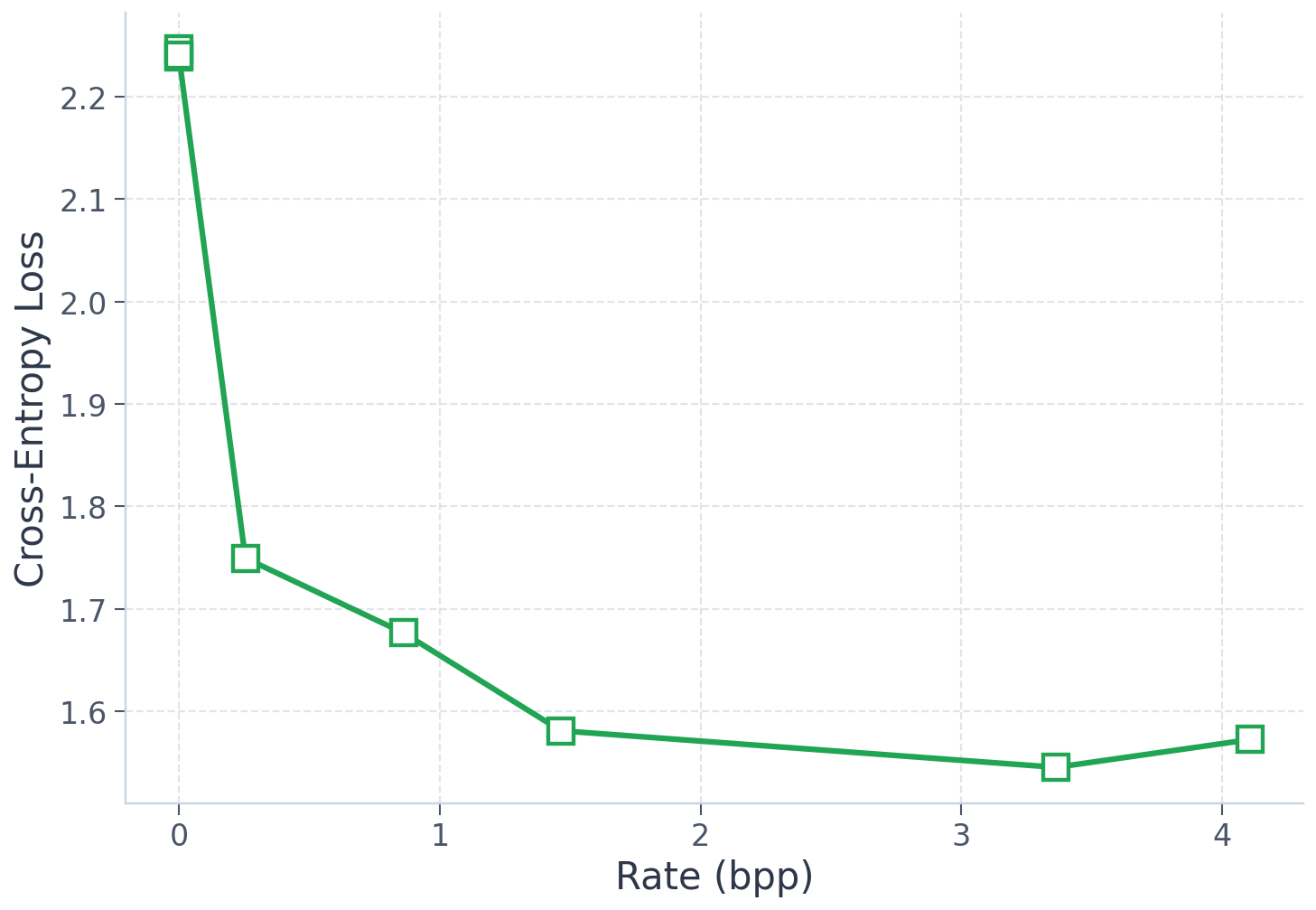}
    \label{fig:Cross_Entropy_Loss_Rate_Denoising_MEC}
}
\hfill
\subfloat[Reconstruction images (SVHN).]{
    \includegraphics[width=0.39\linewidth]{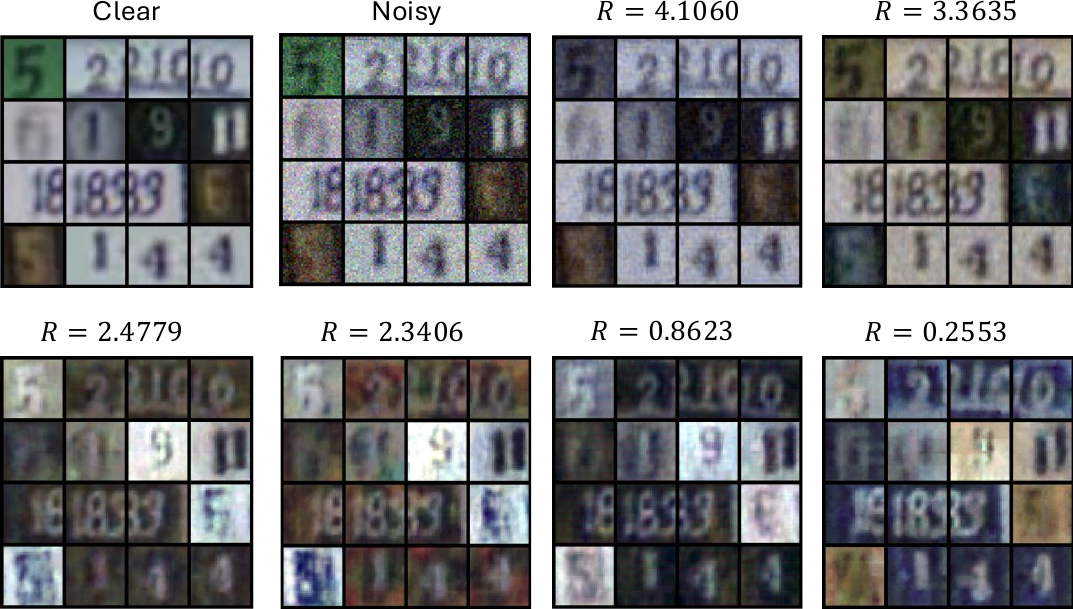}
    \label{fig:Reconstruction_Denosing_MEC}
}
\caption{Experimental results of $4\times$ super-resolution on MNIST dataset and image denoising on the SVHN dataset corrupted by Gaussian noise, $\mathcal{N}(0,\sigma)$ with $\sigma = 25$. Increasing the rate improves classification performance.}
\label{fig:super_resolution_MEC}
\end{figure*}

We consider a cross-domain restoration setting in which the encoder observes
degraded samples $X \sim p_X$ (e.g., noisy or low-resolution) and generates
reconstructions following a target distribution $Y \sim p_Y$ (e.g., clean or
high-resolution). The objective is to compress $X$ while preserving
task-relevant information for a downstream label $S$. The experimental scheme
is shown in Figure~\ref{fig:Scheme_Cross_Domain_MEC}.

To implement Problem~\eqref{prob:oneshot_random_def_classification}, we adopt a
stochastic autoencoder consisting of an encoder $f$, quantizer $Q$, decoder
$g$, classifier $c$, and a WGAN discriminator $d$. The classification
constraint is enforced through the cross-entropy loss
$\text{CE}(S,\hat{S})$, which upper bounds the conditional entropy
$H(S|Y)$~\cite{boudiaf2021unifying_cross_entropy,Wang2024}. The compression
rate is estimated using a learned entropy model:
\[
R = \mathbb{E}_{X \sim p_X}\!\left[-\log \mathbb{P}(Q(f(X,U)))\right],
\]
where $\mathbb{P}(Q(f(X,U)))$ is parameterized as a factorized non-parametric
distribution~\cite{balle2018variational}.

To maximize the coupling objective $I(X;Y)$, we employ the variational lower
bound~\cite{barber2004algorithm,chen2016infogan}
\begin{align*}
I(X;Y)
&\ge I(X;Y)_{\text{lb}} \\
&= H(X)
+ \mathbb{E}_{x \sim p_X,\;
y \sim p_{Y|X}(\cdot|x)}
\big[ \log q_\gamma(x|y) \big],
\end{align*}
where $q_\gamma(x|y)$ is implemented using a neural reconstructor $r$ that
models the reverse mapping from reconstructed samples to degraded inputs.
For a target rate $R$ and shared randomness $U$, the system solves
\begin{align*}
\min_{f,g,Q} \quad
& -\, I_{\mathrm{lb}}(X;Y) \\
\text{s.t.}\quad
& p_{g(Q(f(X,U)))} = p_Y, \\
& \mathbb{E}\!\left[-\log \mathbb{P}(Q(f(X,U)))\right] \le R, \\
& H\!\big(S \mid g(Q(f(X,U)))\big) \le C.
\end{align*}
Letting $\hat{Y} = g(Q(f(X,U)))$, the distributional constraint is enforced
using a WGAN discriminator through the Wasserstein-1
distance~\cite{arjovsky2017wasserstein}. Shared randomness is implemented via
universal quantization~\cite{ziv1985universal,theis2021advantages}. In
practice, we optimize the relaxed objective
\begin{align*}
\mathcal{L} &= - I(X;Y)_{\text{lb}} 
- \lambda_r \log \mathbb{P}(Q(f(X,U))) \\
&+ \lambda_p W_1(p_Y,p_{\hat{Y}}) 
+ \lambda_c \,\textnormal{CE}(S,\hat{S}),
\end{align*}
which balances information preservation, rate regularization,
distribution alignment, and classification performance.

\subsection{Results}

Fig.~\ref{fig:Accuracy_Rate_Resolution_MEC} and
Fig.~\ref{fig:Cross_Entropy_Loss_Rate_Resolution_MEC} show the results for
\(4\times\) super-resolution on MNIST. Increasing the rate improves
classification accuracy and reduces the cross-entropy loss, indicating that
larger bottleneck rates preserve more task-relevant information about the
source. The qualitative reconstructions in
Fig.~\ref{fig:Reconstruction_Resolution_MEC} further show that higher-rate
representations better preserve digit identity and align more closely with the
target high-resolution domain.

Fig.~\ref{fig:Accuracy_Rate_Denoising_MEC} and
Fig.~\ref{fig:Cross_Entropy_Loss_Rate_Denoising_MEC} report the corresponding
results for Gaussian denoising on SVHN. Similar behavior is observed:
increasing the rate improves classification accuracy and decreases the
cross-entropy loss. The qualitative examples in
Fig.~\ref{fig:Reconstruction_Denosing_MEC} show that higher-rate
representations better preserve semantic and visual structure, while low-rate
reconstructions lose fine details. Some denoised SVHN images exhibit mild color
inconsistencies relative to the clean targets. This behavior is consistent with
the information-theoretic nature of the objective, since mutual information is
invariant under invertible transformations, i.e.,
\(I(X;Y)=I(X;f(Y))\). Consequently, preserving task-relevant information alone
may not uniquely determine low-level visual attributes such as color
statistics. 
Such artifacts could be reduced by incorporating additional fidelity regularization terms and careful architectural design.

Overall, the experimental results demonstrate that the proposed
rate- and classification-constrained MEC framework produces more informative
reconstructions as the available rate increases, while maintaining alignment
with the target reconstruction distribution.

\section{Conclusion}

We studied cross-domain lossy compression via rate- and classification-constrained minimum entropy coupling. Motivated by logarithmic loss, the proposed framework uses an information-based objective to couple a degraded source domain with a prescribed target reconstruction domain while preserving task-relevant information. Under common randomness, we derived an equivalent deterministic formulation and obtained closed-form Bernoulli characterizations with and without classification constraints. Experiments on MNIST super-resolution and SVHN denoising show that increasing the available rate improves classification accuracy and yields more informative reconstructions. These results suggest that rate- and classification-constrained MEC provides a useful framework for distribution-constrained and task-aware compression.

\newpage
\bibliographystyle{IEEEtran}
\bibliography{main}

\clearpage
\setcounter{page}{1}
\appendix

\subsection{Proof of Theorem~\ref{thm:oneshot_random}}\label{app:oneshot_random_proof}

Recall the formulation in Definition \ref{def:mecb_randomness}:
\begin{align}
\mathcal{I}_\text{MEC-B-R}(p_X, p_Y, R) =  & \max_{p_{U,X,Z,Y}\in M(p_X,p_Y)} I(X;Y) \nonumber\\
\textnormal{s.t.}\quad & H(Z|U)\le R. \nonumber
\end{align}
where $M(p_X,p_Y)=\{p_U p_X p_{Z|X,U} p_{Y|Z,U}\}$ and
$Q(p_X,p_Y)=\{p_U p_X p_{Y|U,X}: H(Y|U,X)=0\}$.

\noindent\textbf{Converse.}
Consider any \(p_{U,X,Y}\in Q(p_X,p_Y)\) satisfying \(H(Y|U)\le R\), and define \(Z\triangleq Y\).
Then \(p_{U,X,Z,Y}=p_U\,p_X\,\delta_{Z=Y(X,U)}\,\delta_{Y=Z}\in M(p_X,p_Y)\), with
\(H(Z|U)=H(Y|U)\le R\). The joint distribution of \((X,Y)\) is preserved, and therefore \(I(X;Y)\) remains unchanged. Hence
\begin{align*}
\mathcal{I}_\text{MEC-B-R}(p_X, p_Y, R)  \;\ge\; & \max_{p_{U,X,Y}\in Q(p_X,p_Y)} I(X;Y) \nonumber\\
\textnormal{s.t.}\quad 
& H(Y|U,X) = 0, \\
& I(X;U) = 0, \\
& H(Y|U)\le R.
\end{align*}

\noindent\textbf{Upper bound.}
Let \(p_{U,X,Z,Y}\in M(p_X,p_Y)\) be any feasible distribution with \(H(Z|U)\le R\).
By the functional representation lemma \cite{elgamal2011network, li2018strong}, there exist:
(i) a random seed $V_1$ independent of $(U,X)$ and a measurable mapping \(\phi_1\) such that \(Z=\phi_1(U,X,V_1)\) in distribution;
(ii) a random seed $V_2$ independent of $(U,X,V_1)$ and a measurable mapping \(\phi_2\) such that \(Y=\phi_2(U,Z,V_2)\) in distribution.

Let \(U' \triangleq (U,V_1,V_2)\). Then, $Y \;=\; \phi_2\big(U,\,\phi_1(U,X,V_1),\,V_2\big)$ is deterministic given \((U',X)\), so \(H(Y|U',X)=0\) and \((U',X,Y)\in Q(p_X,p_Y)\).
The marginal \((X,Y)\) is preserved, hence \(I(X;Y)\) is unchanged. For the rate term, conditioning reduces entropy, and determinism gives
\[
H(Z|U)\ \ge\ H(Z|U,V_1,V_2)=H(Z|U')\ \ge\ H(Y|U'),
\]
and \(H(Y|U')\le R\).  Therefore,
\begin{align}
\mathcal{I}_\text{MEC-B-R}(p_X, p_Y, R) \;\le\; & \max_{p_{U',X,Y}\in Q(p_X,p_Y)} I(X;Y) \nonumber\\
\textnormal{s.t.}\quad & 
H(Y|U',X) = 0, \nonumber\\
& I(X;U') = 0, \nonumber \\
& H(Y|U')\le R. \nonumber
\end{align}
Since the auxiliary alphabet is unrestricted, we can relabel \(U'\) as \(U\). Combining the two bounds completes the proof.

\subsection{Proof of Theorem~\ref{Bernoulli_radom_IRC}}\label{app:Bernoulli_radom_IRC_proof}

Since $ H(Y|U) = I(X;Y|U) + H(Y|U,X) = I(X;Y|U)$ and $H(Y|U,X)=0$, the reconstruction $Y$ is a deterministic function of $(X,U)$, i.e., $Y = f(X,U)$. Consequently, Problem~\eqref{prob:mecb_randomness_deterministic_convert} reduces to optimizing over the distribution $p_U$:
\begin{align*}
    \mathcal{I}_\text{MEC-B-R}^{(B)}(q_X, q_Y, R) 
    &= \max_{p_U} \quad I(X;Y) \\
    \textnormal{s.t.} \quad 
    & H(Y|U,X) = 0, \\
    & I(X;U) = 0, \\
    & I(X;Y|U) \leq R.
\end{align*}

Since Shannon entropy is defined for discrete random variables, the auxiliary variable $U$ must be selected such that $Y|U=u$ is discrete for each $u$ \cite{liu2022lossy}.  
Let $\mathcal{U} \triangleq \{1,2,\ldots,|\mathcal{Y}|^{|\mathcal{X}|}\}$ denote the index set of all mappings $f_u:\mathcal{X}\to\mathcal{Y}$.  
By the support lemma (Appendix~C, p.~631 of \cite{elgamal2011network}), it suffices to assign positive probability to at most $|\mathcal{Y}|+1$ such mappings.

The optimization can therefore be written as
\begin{align}
   \mathcal{I}_\text{MEC-B-R}^{(B)}(q_X, q_Y, R) 
   &= \max_{p_U} \quad I(X;Y) \nonumber \\
   \textnormal{s.t.} \quad 
   & \sum_{u \in \mathcal{U}} p_U(u) \, I(X;Y | U=u) \leq R, \nonumber \\
   & \sum_{u \in \mathcal{U}} p_U(u) \, P(f_u(X)=y) = q_Y, \forall y \in \mathcal{Y}. \nonumber 
\end{align}

For binary alphabets, the cardinality of $U$ can be restricted to four without loss of optimality. The four distinct mappings from $\{0,1\}$ to $\{0,1\}$ are
$f_1(x) = x$, $f_2(x) = 1 - x$, $f_3(x) = 0$, and $f_4(x) = 1$, where $x \in \{0,1\}$. Accordingly, $Y = X$ if $U=1$, $Y = 1-X$ if $U=2$, $Y = 0$ if $U=3$, and $Y = 1$ if $U=4$. This yields
\begin{align*}
I(X;Y|U) &= \sum_{u \in \mathcal{U}} p_U(u) I(X;Y|U = u) \\
&= \sum_{u \in \mathcal{U}} p_U(u) H(f_u(X)) \!=\! H_b(q_X)(p_U(1) \!+\! p_U(2)), \\
P(Y = 1) &= q_Y = \sum_{u \in \mathcal{U}} p_U(u) P(f_u(X) = y) \\
&= p_U(1)q_X + (1 - q_X)p_U(2) + p_U(4).
\end{align*}
Moreover, for each $x \in \{0,1\}$, we have
\begin{align*}
P(Y=1|X=0) &= p_U(2) + p_U(4), \\
P(Y=1|X=1) &= p_U(1) + p_U(4),
\end{align*}
which implies
\begin{align*}
H(Y|X)
&= \sum_{x \in \mathcal{X}} p_X(x) H(Y|X=x) \\
&= (1-q_X)H_b\bigl(p_U(2)+p_U(4)\bigr) \\
&+ q_X H_b\bigl(p_U(1)+p_U(4)\bigr).
\end{align*}
Since $Y \sim \mathrm{Bern}(q_Y)$, we also have $H(Y)=H_b(q_Y).$
Therefore,
\begin{align*}
I(X;Y)
&= H(Y)-H(Y|X) \\
&= H_b(q_Y) - (1-q_X)H_b\bigl(p_U(2)+p_U(4)\bigr) \\
&- q_X H_b\bigl(p_U(1)+p_U(4)\bigr).
\end{align*}

The resulting optimization is given by (\ref{MEC_B_R}).
\begin{figure*}
\begin{align}
\label{MEC_B_R}
\mathcal{I}_\text{MEC-B-R}(q_X, q_Y, R) \nonumber \\
&= \max_{p_U(1),\, p_U(2),\, p_U(3),\, p_U(4)} \;
H_b(q_Y) - (1-q_X)H_b\bigl(p_U(2)+p_U(4)\bigr)
- q_X H_b\bigl(p_U(1)+p_U(4)\bigr) \\
\textnormal{s.t.} \quad 
& H_b(q_X)\,(p_U(1)+p_U(2)) \leq R, \nonumber\\
& q_X p_U(1) + (1-q_X)p_U(2) + p_U(4) = q_Y, \nonumber\\
& p_U(1)+p_U(2)+p_U(3)+p_U(4) = 1, \nonumber\\
& p_U(1),\, p_U(2),\, p_U(3),\, p_U(4) \geq 0. \nonumber
\end{align} 
\noindent\rule{\textwidth}{0.6pt}
\end{figure*}

From the constraints, we have
\begin{align*}
q_X p_U(1) + (1-q_X)p_U(2) + p_U(4) &= q_Y, \\
p_U(1)+p_U(2)+p_U(3)+p_U(4) &= 1, 
\end{align*}
then, 
\begin{align}
p_U(3) &= 1-q_Y-(1-q_X)p_U(1)-q_Xp_U(2), \nonumber\\
p_U(4) &= q_Y-q_Xp_U(1)-(1-q_X)p_U(2). \label{eq:mecbr_p4_expr}
\end{align}
Substituting \eqref{eq:mecbr_p4_expr} into the objective yields
\begin{align*}
p_U(2)+p_U(4) &= q_Y-q_X\bigl(p_U(1)-p_U(2)\bigr),\\
p_U(1)+p_U(4) &= q_Y+(1-q_X)\bigl(p_U(1)-p_U(2)\bigr).
\end{align*}
Hence, the objective depends on \(\bigl(p_U(1),p_U(2)\bigr)\) only through the difference $d = p_U(1)-p_U(2).$
Define
\begin{align}
F(d)
=
(1-q_X)H_b\bigl(q_Y-q_Xd\bigr)
+q_XH_b\bigl(q_Y+(1-q_X)d\bigr).
\label{eq:mecbr_F_def}
\end{align}
The problem then reduces to minimizing \(F(d)\) over feasible values of \(d\). We shall show that the objective is nonincreasing in \(d\) for \(d\ge 0\). Differentiating \eqref{eq:mecbr_F_def}, we obtain
\begin{align*}
F'(d)
&= -q_X(1-q_X)H_b'\bigl(q_Y-q_Xd\bigr) \\
&+ q_X(1-q_X)H_b'\bigl(q_Y+(1-q_X)d\bigr) \\
&= q_X(1-q_X)
\Bigl[
H_b'\bigl(q_Y \!+\! (1-q_X)d\bigr)
\!-\!
H_b'\bigl(q_Y-q_Xd\bigr)
\Bigr].
\end{align*}
It is noted that \(H_b'(t)=\log\frac{1-t}{t}\) is strictly decreasing on \(t \in (0,1)\). Moreover, $q_Y+(1-q_X)d \ge q_Y-q_Xd$ for all $d\ge 0$. Hence, it follows that $F'(d)\le 0$ and for all $d\ge 0.$ Thus \(F(d)\) is nonincreasing in \(d\). Consequently, minimizing \(F(d)\) is equivalent to maximizing \(d\).

We next show this lemma as follows
\begin{lemma}
Without loss of optimality, we may restrict attention to feasible points with \(p_U(2)=0\).    
\end{lemma}

\begin{proof}
Let \(\bigl(p_U(1),p_U(2),p_U(3),p_U(4)\bigr)\) be any feasible point, and define
\begin{align}
\tilde p_U(1) = p_U(1)-p_U(2), \;\; 
\tilde p_U(2) = 0.
\label{eq:mecbr_tilde12}
\end{align}
Next define \(\tilde p_U(4)\) and \(\tilde p_U(3)\) from the constraints:
\begin{align*}
\tilde p_U(4) &= q_Y-q_X\tilde p_U(1)-(1-q_X)\tilde p_U(2), \label{eq:mecbr_tilde4_def}\\
\tilde p_U(3) &= 1-\tilde p_U(1)-\tilde p_U(2)-\tilde p_U(4). 
\end{align*}
Using \eqref{eq:mecbr_tilde12}, we obtain
\begin{align*}
\tilde p_U(4)
&= q_Y-q_X\bigl(p_U(1)-p_U(2)\bigr) \\
&= \bigl(q_Y-q_Xp_U(1)-(1-q_X)p_U(2)\bigr)+p_U(2) \\
&= p_U(4)+p_U(2) \ge 0,\\
\tilde p_U(3)
&= 1-\bigl(p_U(1)-p_U(2)\bigr)-\tilde p_U(4)
 = 1-p_U(1)-p_U(4) \\
&= p_U(2)+p_U(3) \ge 0.
\end{align*}
Moreover,
\[
\tilde p_U(1)+\tilde p_U(2)=p_U(1)-p_U(2)\le p_U(1)+p_U(2),
\]
so the rate constraint remains satisfied:
\[
H_b(q_X)\bigl(\tilde p_U(1)+\tilde p_U(2)\bigr)
\le H_b(q_X)\bigl(p_U(1)+p_U(2)\bigr)\le R.
\]
Finally, $\tilde p_U(1)-\tilde p_U(2)=p_U(1)-p_U(2)=d,$ so the transformed point attains the same objective value, since the objective depends only on \(d\) through \eqref{eq:mecbr_F_def}. Therefore, every feasible point can be replaced by another feasible point with the same objective value and with \(p_U(2)=0\). Hence, there exists an optimal solution satisfying $p_U^\star(2)=0.$    
\end{proof}

Then, \(d=p_U(1)\ge 0\), and
\begin{align*}
p_U(4) &= q_Y-q_Xp_U(1), \\
p_U(3) &= 1-q_Y-(1-q_X)p_U(1). 
\end{align*}

The rate constraint becomes $p_U(1) \le \frac{R}{H_b(q_X)}.$
The nonnegativity constraints \(p_U(4)\ge 0\) and \(p_U(3)\ge 0\) imply $p_U(1) \le \frac{q_Y}{q_X}$ and $p_U(1) \le \frac{1-q_Y}{1-q_X}.$
Therefore, the largest feasible value of \(p_U(1)\) is $
\min\left\{
\frac{R}{H_b(q_X)},
\frac{q_Y}{q_X},
\frac{1-q_Y}{1-q_X}
\right\}.$
 
\begin{itemize}
    \item If \(q_Y\le q_X\), then $q_Y(1-q_X)\le q_X(1-q_Y),$ which is equivalent to $\frac{q_Y}{q_X}\le \frac{1-q_Y}{1-q_X}.$ Hence, $\min\left\{\frac{q_Y}{q_X},\frac{1-q_Y}{1-q_X}\right\} = \frac{q_Y}{q_X}.$ 

    \item If \(q_Y\ge q_X\), then $q_Y(1-q_X)\ge q_X(1-q_Y),$ which is equivalent to $\frac{q_Y}{q_X}\ge \frac{1-q_Y}{1-q_X}.$ Therefore, $\min\left\{\frac{q_Y}{q_X},\frac{1-q_Y}{1-q_X}\right\} = \frac{1-q_Y}{1-q_X}.$
\end{itemize}

Therefore, the optimal values are given by
\begin{align*}
p_U^\star(1) &= \alpha, &
p_U^\star(2) &= 0, \nonumber\\
p_U^\star(3) &= 1-q_Y-(1-q_X)\alpha, &
p_U^\star(4) &= q_Y-q_X\alpha.
\end{align*}

\begin{figure*}
\begin{align}
\label{MEC_B_R_C}
\mathcal{I}_\text{MEC-B-R-C}^{(B)}(q_X, q_Y, R, C)
&= \max_{p_U(1),\, p_U(2),\, p_U(3),\, p_U(4)} \; H_b(q_Y) - (1-q_X)H_b\bigl(p_U(2)+p_U(4)\bigr) - q_X H_b\bigl(p_U(1)+p_U(4)\bigr) \\
\textnormal{s.t.} \quad 
& H_b(q_X)\,(p_U(1)+p_U(2)) \leq R, \label{rate_common_constraint_1_DRC} \\
& q_X p_U(1) + (1-q_X)p_U(2) + p_U(4) = q_Y, \label{rate_common_constraint_2_DRC} \\
& (p_U(1)+p_U(2)) H_b(q_{S_1}) + (p_U(3)+p_U(4)) H_b(m) \leq C, \label{rate_common_constraint_3_DRC}\\
& p_U(1)+p_U(2)+p_U(3)+p_U(4) = 1, \label{rate_common_constraint_4_DRC}\\
& p_U(1),\, p_U(2),\, p_U(3),\, p_U(4) \geq 0. \label{rate_common_constraint_5_DRC}
\end{align}   
\noindent\rule{\textwidth}{0.6pt}
\end{figure*}

And, 
\begin{align*}
\mathcal{I}_\text{MEC-B-R}^{(B)}(q_X,q_Y,R)
&=
H_b(q_Y) - (1-q_X)H_b\bigl(q_Y-q_X\alpha\bigr)
\\
&-
q_XH_b\bigl(q_Y+(1-q_X)\alpha\bigr),
\end{align*}
where 
\begin{align*}
\alpha
=
\begin{cases}
\min\left\{\dfrac{R}{H_b(q_X)},\, \dfrac{q_Y}{q_X}\right\},
& q_Y \le q_X,\\[2ex]
\min\left\{\dfrac{R}{H_b(q_X)},\, \dfrac{1-q_Y}{1-q_X}\right\},
& q_Y \ge q_X.
\end{cases}
\end{align*}
This completes the proof.

\subsection{Proof of Theorem~\ref{Oneshot_Bernoulli_radom_IRC_classifiction}}
\label{app:Oneshot_Bernoulli_radom_IRC_classifiction_proof}

Follows the Appendix~\ref{app:Bernoulli_radom_IRC_proof}, since $H(Y|U,X)=0$, the reconstruction $Y$ is a deterministic function of $(X,U)$, i.e., $Y=f(X,U)$, and the problem~\eqref{prob:oneshot_random_def_classification} reduces to optimizing over $p_U$:
\begin{align*}
    \mathcal{I}_\text{MEC-B-R-C}^{(B)}(q_X, q_Y, R, C)
    &= \max_{p_U} \quad I(X;Y) \\
    \textnormal{s.t.} \quad 
    & H(Y|U,X) = 0, \\
    & I(X;U) = 0, \\
    & I(X;Y|U) \leq R, \\
    & H(S|Y) \leq C.
\end{align*}

As in the $\mathcal{I}_\text{MEC-B-R}^{(B)}(q_X, q_Y, R)$ case, it suffices to consider $|\mathcal{U}|\le 4$, corresponding to the four deterministic mappings $f_u:\{0,1\}\to\{0,1\}$. 
Thus, all expressions for $I(X;Y|U)$, $P(Y=1)$, and $H(Y|X)$ remain unchanged, yielding
\begin{align*}
I(X;Y)
&= H_b(q_Y) - (1-q_X)H_b\bigl(p_U(2)+p_U(4)\bigr) \\
&- q_X H_b\bigl(p_U(1)+p_U(4)\bigr).
\end{align*}

The rate and marginal constraints are identical to the previous case. The key difference lies in the classification constraint.
By the data-processing inequality \cite{cover1999elements} and the Markov chain $S \leftrightarrow X \leftrightarrow Y$,
\begin{align*}
H(S|Y) \;\geq\; H(S|X) = H(X \oplus S_1 | X) = H(S_1) = H_b(q_{S_1}),
\end{align*}
and therefore feasibility requires $C \ge H_b(q_{S_1})$.
We evaluate the classification loss under each mapping:
\begin{itemize}
    \item For $U=1$ or $U=2$: $Y$ is a bijective function of $X$, hence $H(S|Y,U=u)=H(S|X)=H_b(q_{S_1}).$
    \item For $U=3$ or $U=4$: $Y$ is constant, and thus independent of $X$. Hence, 
    \begin{align*}
    P(S=0) = (1-q_X)(1-q_{S_1}) + q_X q_{S_1}.    
    \end{align*}
    And, 
    \begin{align*}
    H(S|Y,U=u) =H_b\!\big((1-q_X)(1-q_{S_1}) + q_X q_{S_1}\big).
    \end{align*}
\end{itemize}

Combining these, we obtain
\begin{align*}
H(S | Y) 
&= \sum_{u \in \mathcal{U}} p_U(u)\, H(S | f_u(X)) \\
&= (p_U(1)+p_U(2)) H_b(q_{S_1}) \\
&+ (p_U(3)+p_U(4)) H_b\!\big((1-q_X)(1-q_{S_1}) + q_X q_{S_1}\big).
\end{align*}

Let $m = (1-q_X)(1-q_{S_1}) + q_X q_{S_1}$, so that
\begin{align*}
H(S | Y)
&= (p_U(1)+p_U(2)) H_b(q_{S_1}) \\
&+ (p_U(3)+p_U(4)) H_b(m).
\end{align*}

The final optimization problem is represented as (\ref{MEC_B_R_C}). 
We can also write
\begin{align*}
p_U(3) &= (1-q_Y) - (1-q_X) [p_U(1)+p_U(2)] \\
&- (2q_X-1)\,p_U(2). \\
p_U(4) &= q_Y - q_X [p_U(1)+p_U(2)] + (2q_X-1)\,p_U(2),
\end{align*}

Let also $d = p_U(1)-p_U(2).$
Similar to the proof in Appendix~\ref{app:Bernoulli_radom_IRC_proof}, we obtain
\begin{align*}
p_U(2)+p_U(4) &= q_Y-q_Xd, \nonumber\\
p_U(1)+p_U(4) &= q_Y+(1-q_X)d.
\end{align*}
Hence,
\begin{align*}
&F(d) = (1-q_X)H_b\bigl(p_U(2)+p_U(4)\bigr) + q_X H_b\bigl(p_U(1)+p_U(4)\bigr) \\
&=
(1-q_X)H_b(q_Y-q_Xd) 
+ q_XH_b(q_Y+(1-q_X)d).
\end{align*}

Therefore, the objective function in~(\ref{MEC_B_R_C}) depends only on \(d\).
From the proof of Theorem~\ref{Bernoulli_radom_IRC}, it also follows that
\(F(d)\) is nonincreasing for \(d\ge 0\) and nondecreasing for \(d<0\).
We therefore consider the following two cases. \\

\noindent
\textbf{PART I: \(d\ge 0\) or $p_U(1) \geq p_U(2)$.}

Since \(F(d)\) is nonincreasing for \(d\ge 0\), minimizing \(F(d)\) over the
feasible nonnegative values of \(d\) is equivalent to selecting the largest
feasible value of \(d\). Thus, the minimum in this case is attained when
$d = p_U(1) > 0$ and $p_U(2)=0.$ Hence,
\begin{align*}
p_U(3)&=1-q_Y-(1-q_X)p_U(1), \\
p_U(4)&=q_Y-q_Xp_U(1).
\end{align*}
This branch is feasible only if
\[
p_U(1)
\le
\min\left\{
\frac{q_Y}{q_X},
\frac{1-q_Y}{1-q_X}
\right\}.
\]

An analytical solution to~(\ref{MEC_B_R_C}) can be obtained using the
Karush--Kuhn--Tucker (KKT) conditions~\cite{boyd2004convex}. We characterize
the solution by considering all possible active/inactive combinations of the
rate and classification constraints.

\textbf{Case 1.} Constraint~\eqref{rate_common_constraint_1_DRC} is active and constraint~\eqref{rate_common_constraint_3_DRC} is inactive.  

Using $p_U(2) = 0$ and the fact that~\eqref{rate_common_constraint_1_DRC} is active, we obtain
\begin{align*}
    &R = H_b(q_X)(p_U(1) + p_U(2)) = H_b(q_X)\,p_U(1) \\
    & p_U(1) = \frac{R}{H_b(q_X)}.
\end{align*}
Therefore,
\begin{align*}
&\frac{R}{H_b(q_X)}
\le
\min\left\{
\frac{q_Y}{q_X},
\frac{1-q_Y}{1-q_X}
\right\} \\
& R \leq H_b(q_X) \min\left\{
    \frac{q_Y}{q_X},
    \frac{1-q_Y}{1-q_X}
    \right\}.   
\end{align*}
Moreover,
\begin{align*}
    &\mathcal{I}_\text{MEC-B-R-C}^{(B)}(q_X, q_Y, R, C) \\
    &\!=\! H_b(q_Y) - (1-q_X)H_b \left[q_Y-q_X\frac{R}{H_b(q_X)} \right] \\
    &- q_XH_b\left[q_Y+(1-q_X)\frac{R}{H_b(q_X)}\right]. 
\end{align*}
The corresponding optimizer is
\begin{align*}
p_U^\star(1) &= \tfrac{R}{H_b(q_X)}, \\
p_U^\star(2) &= 0, \\
p_U^\star(3) &= \tfrac{-(1 - q_X) R}{H_b(q_X)} + 1 - q_Y, \\
p_U^\star(4) &= \tfrac{-q_X R}{H_b(q_X)} + q_Y.
\end{align*}
The classification constraint~\eqref{rate_common_constraint_3_DRC} is inactive when
\begin{align*}
&(p_U(1) + p_U(2)) H_b(q_{S_1}) + (p_U(3) + p_U(4)) H_b(m) < C, \\
&C > \frac{R (H_b(q_{S_1}) - H_b(m))}{H_b(q_X)}  + H_b(m).
\end{align*}

\textbf{Case 2.} Constraint~\eqref{rate_common_constraint_3_DRC} is active and constraint~\eqref{rate_common_constraint_1_DRC} is inactive.  

The constraint~\eqref{rate_common_constraint_3_DRC} is active if
\begin{align*}
(p_U(1) + p_U(2)) H_b(q_{S_1}) + (p_U(3) + p_U(4)) H_b(m) = C.
\end{align*}
Using constraint~\eqref{rate_common_constraint_4_DRC} and $p_U(2) = 0$, we have
\begin{align*}
&(p_U(1) + p_U(2)) H_b(q_{S_1}) + (1 - p_U(1) - p_U(2)) H_b(m) = C, \\ 
&p_U(1) = \frac{C - H_b(m)}{H_b(q_{S_1}) - H_b(m)}.
\end{align*}
Hence,
\begin{align*}
&\frac{C-H_b(m)}{H_b(q_{S_1})-H_b(m)}
\le
\min\left\{
\frac{q_Y}{q_X},
\frac{1-q_Y}{1-q_X}
\right\} \\ 
& C \geq H_b(m) - (H_b(m) - H_b(q_{S_1}))\min\left\{
    \frac{q_Y}{q_X},
    \frac{1-q_Y}{1-q_X}
    \right\}.
\end{align*}

We first state the following auxiliary fact, which will be used below.

\begin{lemma}\label{lem:entropy_order}
We have
\[
H_b(m)\;\ge\; H_b(q_{S_1}),
\]
with equality if and only if $q_X\in\{0,1\}$ or $q_{S_1}=\tfrac12$.
\end{lemma}

\begin{proof}[Proof of Lemma \ref{lem:entropy_order}]
The expression for $m$ can be rewritten as
\begin{align*}
&m = (1-q_X)(1-q_{S_1}) + q_X q_{S_1} 
= \tfrac{1}{2} + \big(q_X - \tfrac{1}{2}\big)\big(2q_{S_1}-1\big),\\
&m-\tfrac12  = \bigl(q_X-\tfrac12\bigr)\bigl(2q_{S_1}-1\bigr).
\end{align*}
Therefore,
\begin{align*}
\bigl|m-\tfrac12\bigr|
\;=\; 2\,\bigl|q_X-\tfrac12\bigr|\,\bigl|q_{S_1}-\tfrac12\bigr|
\;\le\; \bigl|q_{S_1}-\tfrac12\bigr|,
\end{align*}
since $|q_X-\tfrac12|\le \tfrac12$ for $q_X\in[0,1]$.

Because the binary entropy function is maximized at $\tfrac12$ and decreases
strictly as $|p-\tfrac12|$ increases, the above inequality implies
$H_b(m)\ge H_b(q_{S_1})$. Equality holds if and only if
$|q_X-\tfrac12|=\tfrac12$, i.e., $q_X\in\{0,1\}$, or
$|q_{S_1}-\tfrac12|=0$, i.e., $q_{S_1}=\tfrac12$.
\end{proof}

Lemma~\ref{lem:entropy_order} justifies the sign of denominators involving
$H_b(m)-H_b(q_{S_1})$. Thus, $H_b(m) \geq C \geq H_b(q_{S_1})$.
Furthermore,
\begin{align*}
    &\mathcal{I}_\text{MEC-B-R-C}^{(B)}(q_X, q_Y, R, C) \\
    & \!=\! H_b(q_Y) - (1-q_X)H_b\left[q_Y-q_X\frac{C - H_b(m)}{H_b(q_{S_1}) - H_b(m)}\right] \\
    &- q_XH_b\left[q_Y+(1-q_X)\frac{C - H_b(m)}{H_b(q_{S_1}) - H_b(m)}\right]. 
\end{align*}

The corresponding optimizer is
\begin{align*}
p_U^\star(1) &= \tfrac{C - H_b(m)}{H_b(q_{S_1}) - H_b(m)}, \\
p_U^\star(2) &= 0, \\
p_U^\star(3) &= \tfrac{-(1 - q_X)(C - H_b(m))}{H_b(q_{S_1}) - H_b(m)} + 1 - q_Y, \\
p_U^\star(4) &= \tfrac{-q_X(C - H_b(m))}{H_b(q_{S_1}) - H_b(m)} + q_Y.
\end{align*}

The rate constraint~\eqref{rate_common_constraint_1_DRC} is inactive when
\begin{align*}
&H_b(q_X)(p_U(1) + p_U(2)) < R \\
& C < \frac{R (H_b(q_{S_1}) - H_b(m))}{H_b(q_X)} + H_b(m).
\end{align*}

\textbf{Case 3.} Both constraints~\eqref{rate_common_constraint_1_DRC} and~\eqref{rate_common_constraint_3_DRC} are active.  

From Case 2, the classification constraint~\eqref{rate_common_constraint_3_DRC} is active when
\[
p_U(1) + p_U(2) = \frac{C - H_b(m)}{H_b(q_{S_1}) - H_b(m)}.
\]
In this case,
\begin{align*}
    &\mathcal{I}_\text{MEC-B-R-C}^{(B)}(q_X, q_Y, R, C) \\
    & \!=\! H_b(q_Y) - (1-q_X)H_b\left[q_Y-q_X\frac{C - H_b(m)}{H_b(q_{S_1}) - H_b(m)}\right] \\
    &- q_XH_b\left[q_Y+(1-q_X)\frac{C - H_b(m)}{H_b(q_{S_1}) - H_b(m)}\right].
\end{align*}

The rate constraint~\eqref{rate_common_constraint_1_DRC} is active if
\[
C = \frac{R (H_b(q_{S_1}) - H_b(m))}{H_b(q_X)} + H_b(m).
\]

\textbf{Case 4.} Both constraints~\eqref{rate_common_constraint_1_DRC} and~\eqref{rate_common_constraint_3_DRC} are inactive.  

When \( C > H_b(q_S) \), the classification constraint~\eqref{rate_common_constraint_3_DRC} is inactive. If, in addition, the rate is sufficiently large, namely \( R > H_b(q_X) \), then the rate constraint~\eqref{rate_common_constraint_1_DRC} is also inactive. The problem then reduces to the unconstrained maximum-information coupling between
\(X\sim\mathrm{Bern}(q_X)\) and \(Y\sim\mathrm{Bern}(q_Y)\) under the prescribed marginals.

Let $\theta = P(X=1,Y=1).$ Since \(X\sim\mathrm{Bern}(q_X)\) and \(Y\sim\mathrm{Bern}(q_Y)\), the joint probabilities are
\begin{align*}
P(X=1,Y=1) &= \theta,\\
P(X=1,Y=0) &= q_X-\theta,\\
P(X=0,Y=1) &= q_Y-\theta,\\
P(X=0,Y=0) &= 1-q_X-q_Y+\theta.
\end{align*}
Here, $\theta$ satisfies the Fréchet--Hoeffding bounds~\cite[Sec.~2.5]{Nelsen2006},
\[
\max\{0,q_X+q_Y-1\}\le \theta\le \min\{q_X,q_Y\}.
\]
Moreover,
\[
H(Y|X)
=
(1-q_X)H_b\left(\frac{q_Y-\theta}{1-q_X}\right)
+
q_XH_b\left(\frac{\theta}{q_X}\right).
\]
For \(0<q_X,q_Y\le \frac12\), the minimum of $H(Y|X)$ is attained at the
positive-dependence endpoint $\theta^*=\min\{q_X,q_Y\}$~\cite{embrechts2001modelling}.
Therefore,
\begin{align*}
&\mathcal{I}_\text{MEC-B-R-C}^{(B)}(q_X,q_Y,R,C) \\
&= H_b(q_Y)- (1-q_X)H_b\left(\frac{q_Y- \min\{q_X,q_Y\}}{1-q_X}\right) \\
&-
q_XH_b\left(\frac{\min\{q_X,q_Y\}}{q_X}\right) \\
&=
\begin{cases}
\displaystyle H_b(q_Y)-q_XH_b\left(\frac{q_Y}{q_X}\right), \;\;
 q_Y\le q_X\\
\displaystyle H_b(q_Y) -(1-q_X)H_b\left(\frac{q_Y-q_X}{1-q_X}\right), \;\;
 q_Y\ge q_X.
\end{cases}    
\end{align*}

\noindent
\textbf{PART II: \(d<0\) or $p_U(1) < p_U(2)$.}

Since \(F(d)\) is nondecreasing for \(d<0\), minimizing \(F(d)\) over the
feasible negative values of \(d\) is equivalent to selecting the smallest
feasible value of \(d\). Hence, the minimum in this case is attained when
$d = -p_U(2) < 0$ and $p_U(1)=0$. Then,
\begin{align*}
p_U(3)=1-q_Y-q_Xp_U(2), \\
p_U(4)=q_Y-(1-q_X)p_U(2).
\end{align*}
Therefore, this branch is feasible only if
\[
p_U(2)
\le
\min\left\{
\frac{q_Y}{1-q_X},
\frac{1-q_Y}{q_X}
\right\}.
\]

We again consider all possible active/inactive combinations of the rate and
classification constraints.

\textbf{Case 1.} Constraint~\eqref{rate_common_constraint_1_DRC} is active and constraint~\eqref{rate_common_constraint_3_DRC} is inactive.  

Using $p_U(1) = 0$ and the fact that~\eqref{rate_common_constraint_1_DRC} is active, we obtain
\begin{align*}
    &R  = H_b(q_X)(p_U(1) + p_U(2)) = H_b(q_X)\,p_U(2)  \\
    & p_U(2) = \frac{R}{H_b(q_X)}.
\end{align*}
Thus,
\begin{align*}
&\frac{R}{H_b(q_X)}
\le
\min\left\{
\frac{q_Y}{1-q_X},
\frac{1-q_Y}{q_X}
\right\} \\
& R \leq H_b(q_X) \min\left\{
    \frac{q_Y}{1-q_X},
    \frac{1-q_Y}{q_X}
    \right\}.
\end{align*}
Moreover,
\begin{align*}
    &\mathcal{I}_\text{MEC-B-R-C}^{(B)}(q_X, q_Y, R, C) \\
    & \!=\! H_b(q_Y) - (1-q_X)H_b\left[q_Y + q_X\frac{R}{H_b(q_X)}\right] \\
    &- q_XH_b\left[q_Y - (1-q_X)\frac{R}{H_b(q_X)}\right]. 
\end{align*}
The corresponding optimizer is
\begin{align*}
p_U^\star(1) &= 0, \\
p_U^\star(2) &= \tfrac{R}{H_b(q_X)}, \\
p_U^\star(3) &= -\tfrac{q_X R}{H_b(q_X)} + 1 - q_Y, \\
p_U^\star(4) &= \tfrac{(q_X - 1) R}{H_b(q_X)} + q_Y.
\end{align*}
The classification constraint~\eqref{rate_common_constraint_3_DRC} is inactive when
\begin{align*}
&(p_U(1) + p_U(2)) H_b(q_{S_1}) + (p_U(3) + p_U(4)) H_b(m) < C, \\
&C > \frac{R (H_b(q_{S_1}) - H_b(m))}{H_b(q_X)}  + H_b(m).
\end{align*}

\textbf{Case 2.} Constraint~\eqref{rate_common_constraint_3_DRC} is active and constraint~\eqref{rate_common_constraint_1_DRC} is inactive.  

The constraint~\eqref{rate_common_constraint_3_DRC} is active if
\begin{align*}
(p_U(1) + p_U(2)) H_b(q_{S_1}) + (p_U(3) + p_U(4)) H_b(m) = C.
\end{align*}
Using constraint~\eqref{rate_common_constraint_4_DRC} and $p_U(1) = 0$, we have
\begin{align*}
&(p_U(1) + p_U(2)) H_b(q_{S_1}) + (1 - p_U(1) - p_U(2)) H_b(m) = C, \\ 
&p_U(1) + p_U(2) = \frac{C - H_b(m)}{H_b(q_{S_1}) - H_b(m)}.
\end{align*}
Therefore,
\begin{align*}
&\frac{C-H_b(m)}{H_b(q_{S_1})-H_b(m)}
\le
\min\left\{
\frac{q_Y}{1-q_X},
\frac{1-q_Y}{q_X}
\right\} \\
& C \geq H_b(m) - (H_b(m) - H_b(q_{S_1}))\min\left\{
    \frac{q_Y}{1-q_X},
    \frac{1-q_Y}{q_X}
    \right\}.
\end{align*}

It follows that
\begin{align*}
    &\mathcal{I}_\text{MEC-B-R-C}^{(B)}(q_X, q_Y, R, C) \\
    & \!=\! H_b(q_Y) - (1-q_X)H_b\left[q_Y + q_X\frac{C - H_b(m)}{H_b(q_{S_1}) - H_b(m)}\right] \\
    &- q_XH_b\left[q_Y - (1-q_X)\frac{C - H_b(m)}{H_b(q_{S_1}) - H_b(m)}\right]. 
\end{align*}
The corresponding optimizer is
\begin{align*}
p_U^\star(1) &= 0, \\
p_U^\star(2) &= \tfrac{C - H_b(m)}{H_b(q_{S_1}) - H_b(m)}, \\
p_U^\star(3) &= \tfrac{-q_X(C - H_b(m))}{H_b(q_{S_1}) - H_b(m)} + 1 - q_Y, \\
p_U^\star(4) &= \tfrac{(q_X - 1)(C - H_b(m))}{H_b(q_{S_1}) - H_b(m)} + q_Y.
\end{align*}

The rate constraint~\eqref{rate_common_constraint_1_DRC} is inactive when
\begin{align*}
&H_b(q_X)(p_U(1) + p_U(2)) < R \\
& C < \frac{R (H_b(q_{S_1}) - H_b(m))}{H_b(q_X)} + H_b(m).
\end{align*}

\textbf{Case 3.} Both constraints~\eqref{rate_common_constraint_1_DRC} and~\eqref{rate_common_constraint_3_DRC} are active.  

From Case 2, the classification constraint~\eqref{rate_common_constraint_3_DRC} is active when
\[
p_U(1) + p_U(2) = \frac{C - H_b(m)}{H_b(q_{S_1}) - H_b(m)}.
\]
In this case,
\begin{align*}
    &\mathcal{I}_\text{MEC-B-R-C}^{(B)}(q_X, q_Y, R, C) \\
    & \!=\! H_b(q_Y) - (1-q_X)H_b\left[q_Y + q_X\frac{C - H_b(m)}{H_b(q_{S_1}) - H_b(m)}\right] \\
    &- q_XH_b\left[q_Y - (1-q_X)\frac{C - H_b(m)}{H_b(q_{S_1}) - H_b(m)}\right]. 
\end{align*}

The rate constraint~\eqref{rate_common_constraint_1_DRC} is active if
\[
C = \frac{R (H_b(q_{S_1}) - H_b(m))}{H_b(q_X)} + H_b(m).
\]

\textbf{Case 4.} Both constraints~\eqref{rate_common_constraint_1_DRC} and~\eqref{rate_common_constraint_3_DRC} are inactive.  

When \( C > H_b(q_S) \), the classification constraint~\eqref{rate_common_constraint_3_DRC} is inactive. If the rate is also sufficiently large, namely \( R > H_b(q_X) \), then the rate constraint~\eqref{rate_common_constraint_1_DRC} is inactive as well. Similar to Part I, the maximum achievable objective value in this case is
\begin{align*}
&\mathcal{I}_\text{MEC-B-R-C}^{(B)}(q_X,q_Y,R,C) \\
&= (1-q_X)H_b\left(\frac{q_Y- \min\{q_X,q_Y\}}{1-q_X}\right) \\
&+
q_XH_b\left(\frac{\min\{q_X,q_Y\}}{q_X}\right).
\end{align*}

Combining the cases in Parts I and II gives the closed-form expression for
\(\mathcal{I}_\text{MEC-B-R-C}^{(B)}(q_X, q_Y, R, C)\), as stated in
Theorem~\ref{Oneshot_Bernoulli_radom_IRC_classifiction}.

\end{document}